\documentclass{kais}
\usepackage[utf8]{inputenc}
\usepackage{amssymb}
\usepackage{graphicx}

\usepackage{amsmath}
\usepackage{algorithm}
\usepackage{algpseudocode}
\usepackage{booktabs}
\usepackage{cite}
\usepackage{multirow}
\usepackage{chngpage}

\usepackage{url}

\newcommand{\scrC}{\ensuremath{\mathcal{C}}}
\newcommand{\scrD}{\ensuremath{\mathcal{D}}}
\newcommand{\scrF}{\ensuremath{\mathcal{F}}}

\newcommand{\scrA}{\ensuremath{\mathcal{A}}}

\newtheorem{problem}{Problem}
\newtheorem{proposition}{Proposition}
\renewcommand{\algorithmicrequire}{\textbf{Input: }}
\renewcommand{\algorithmicensure}{\textbf{Output: }}

\DeclareMathOperator*{\argmin}{arg\,min}
\DeclareMathOperator*{\argmax}{arg\,max}

\begin{document}


\title{Discovering Compressing Serial Episodes from Event Sequences}
\author[Ibrahim. A et al]{Ibrahim A$^1$, Shivakumar Sastry$^2$, P.S. Sastry$^1$
\\ $^1$Department of Electrical Engineering, Indian Insitute of Science, Bangalore\\ $^2$Department of Electrical and Computer Engineering
Unversity of Akron}

\maketitle
\begin{abstract}
Most pattern mining methods yield a large number of frequent patterns and isolating a small, relevant subset of patterns is a challenging problem of
current interest. In this paper we address this problem in the context of discovering frequent episodes from symbolic time series data. Motivated 
by the Minimum Description Length principle, we formulate the problem of selecting relevant subset of patterns as one of searching for a subset of
patterns that achieves best data compression. We present algorithms for discovering small sets of relevant non-redundant episodes that achieve good
data compression. The algorithms employ a novel encoding scheme and use serial episodes with inter-event constraints as the patterns. We present
extensive simulation studies with both synthetic and real data, comparing our method with the existing schemes such as GoKrimp and SQS. We also
demonstrate the effectiveness of these algorithms on event sequences from a composable conveyor system; this system represents a new application area
where use of frequent patterns for compressing the event sequence is likely to be important for decision-support and control. 
\end{abstract}
\begin{keywords}
 Frequent Episodes, Serial Episodes, Mining Event Sequences, Discovering Compressing 
Patterns, MDL, Inter-Event-Time constraints.
\end{keywords}


\section{Introduction}
Frequent pattern mining is an important problem in the area of data mining that has diverse applications in a variety of
domains~\cite{han2007frequent}. Even though many algorithms have been proposed for frequent pattern mining, most of these methods produce a large
number of frequent patterns. In addition, the patterns found are often redundant in the sense that many patterns are very similar. The redundancy and
the large volume of the patterns discovered makes it difficult to use the mined patterns to gain useful insights into the data or to use them to
extract rules which are effective for prediction, classification etc., in the application domain. Thus,  finding a small set of non-redundant,
relevant and informative patterns that succinctly characterize the data, is an important problem of current interest. 

There are many methods that are proposed for reducing the number of extracted frequent patterns. 
Many such methods concentrate on eliminating patterns that are deemed to be non-informative given the other frequent patterns. For example, in the
context of transaction datasets, concepts such as closed~\cite{pasquier1999discovering,wang2003closet+}, non-derivable~\cite{calders2002mining} and
maximal~\cite{lin2002pincer,burdick2005mafia} itemsets were suggested to reduce the number of frequent itemsets extracted. Similarly,  closed
sequential patterns were proposed for sequence datasets~\cite{yan2003clospan,wang2004bide,casas2005summarizing}. Even though such methods result in
some reduction in the number of patterns returned by the algorithm, the number of patterns still remains substantial. Also, the redundancy in the
final set of patterns is, often, still large.

Recently, there have been other efforts for finding a small set of informative patterns that best describes the data. For example, 
\cite{chandola2007summarization} proposes a method for summarization of transaction datasets based on some ideas from information theory. They propose
a method of selecting a subset of frequent itemsets to achieve a good lossy summarization of the database. Here each transaction is summarized by one
itemset with as little loss of information as possible. In~\cite{wang2006efficiently}, which also proposes a lossy summarization, each
transaction is covered, partially, by the largest frequent itemset. In contrast to these methods,  ~\cite{siebes2006item,vreeken2011krimp} propose
lossless summarization of transaction datasets using the Minimum Description Length~(MDL) principle. A related approach called {\it Tiling} was used
by~\cite{geerts2004tiling,xiang2008succinct}, again for a lossless summarization of the data.  

In this paper, we address the problem of discovering a set of patterns that can achieve succinct lossless representation of temporal sequence data. We
present algorithms that discover a small set of relevant patterns (which are special forms of serial episodes) which summarize the data well. We use
the  MDL principle \cite{grunwald2007minimum} to define what we mean by summarizing the data well. The basic idea is that a set of patterns
characterizes or summarizes the data sequence well, if the set of patterns can be used as a model to encode the data to achieve good compression. 

As mentioned above, the MDL principle has been used earlier to obtain relevant and non-redundant subsets of frequent patterns. The idea was first
explored by the Krimp algorithm \cite{vreeken2011krimp} in the context of transaction data. This algorithm selects a subset of frequent itemsets
which, when used for encoding the database, achieves good compression. Each selected itemset is assigned a code with shorter code lengths assigned to
higher frequency itemsets. The algorithm tries to encode each transaction with the codes of itemsets which have minimal code lengths and which cover
maximum number of items.

Similar strategies have been proposed for sequence data also~\cite{tatti2012long,lam2012mining,lam2014mining}. For sequential data, unlike in the case
of transaction data, the temporal ordering is important and this presents additional complications while encoding the data. For example, consider a
single transaction, $t = ABCD$ from a transaction database and two itemsets $AC$ and $BD$. The codewords for $AC$ and $BD$ can encode the transaction
$t$ (since the transaction is just a set of items). Now consider a sequence $s = ABCD$ and  two  serial episodes $A \rightarrow  C$  and $B
\rightarrow D$.  Even  though the  occurrences of $A \rightarrow C$ and $B \rightarrow D$ would cover the  sequence $s$, this information alone is
insufficient for encoding the sequence. Since the order of events is important in sequential data, in order to get back the exact sequence, one needs
to specify where exactly the occurrences of the episodes  happen in the sequence. For example, we need to know that the $A$ and $C$ in the occurrence
of  $A \rightarrow  C$ are not contiguous and that there is a $B$ in the gap between them. One  needs to have some way of taking care  of such gaps
while encoding the data with the occurrences of some frequent episodes. In general, the events in a sequence constituting an episode occurrence need
not be contiguous, and different occurrences can have arbitrary temporal overlaps. An encoding scheme should be able to properly take care of this. 

The previous approaches for using the MDL principle to summarize sequence data~\cite{tatti2012long, lam2012mining, lam2014mining}, explicitly record
such gaps while encoding data, thus significantly increasing the encoding length. While the methods presented in~\cite{tatti2012long, lam2014mining}
consider only  sequential data without time stamps, the method in~\cite{lam2012mining} does encode event sequences with time stamps also; but the
encoding scheme needs to individually encode each event time stamp. In some cases, the resulting encoding may become even longer than the raw
data~\cite{lam2012mining}. For the problem of identifying a relevant subset of frequent patterns, we are using the encoded length (of the data encoded
with a subset of patterns), only as a figure of merit to compare different subsets. Hence, the fact of the encoding length becoming more than the raw
data is, per se, not disallowed. However, the underlying philosophy of MDL principle suggests that one needs a good level of data compression to have
confidence in a model. For example, even if the sequence data is {\em iid} noise and has no temporal structure, there would be some subset of patterns
that would achieve lower encoded data length than other subsets. However, one expects that even the best such subset here would not achieve any
appreciable level of data compression, thus suggesting that there are no significant temporal regularities in the data. On the other hand, for a
sequence with significant temporal regularities, one expects good compression of the data sequence, if the method is able to discover the best
temporal patterns and encode the sequence with them. In general, if we can discover some long episodes which occur many times, then their occurrences
can encode many events in the data sequence thus giving rise to the possibility of data compression. 

In this paper, we consider summarizing event sequences (having time stamps on events) using a pattern class consisting of serial episodes with fixed
inter-event times. We present algorithms for discovering a small subset of relevant frequent episodes that result in good compression of the data
sequence. The novelty of our approach is that, in contrast to the existing schemes in \cite{tatti2012long,lam2012mining,lam2014mining}, our method
does not need to explicitly encode gaps in episode occurrences and the encoding scheme is such that we can retrieve the full data sequence with the
time stamps on events, from the encoded sequence. The encoding of the data consists of only the start times of occurrences of various episodes; the
gaps are determined from the fixed inter-event time constraints of the episodes. We show through simulations that our method results
in better data compression. We also show, through empirical experiments, that the episodes that result in good data compression are also highly
relevant for the dataset. 

We also illustrate the benefits of our algorithm using an application, where it is important to both find relevant patterns and achieve good data
compression.  We consider streams of sensor-data from a composable conveyor system (CCS)~\cite{asrr2009,sks2006} that is useful for materials
handling. In this system, several conveying units are dynamically composed to achieve the application objectives; consequently, utilizing the data
streams to diagnose or reconfigure the system is important. The data consists of a sequence of predefined events such as, {\em package entered a
unit}, {\em package exited a unit}, {\em package arrived at an input port}, etc. Such events  occur at various units in the conveyor system during its
routine operation.  On this data stream, frequent serial episodes represent the routes (sequence of units) over which packages were transported in the
conveyor system. The inter-event times corresponds to the various physical constraints such as time required for a package to move through a specific
unit, the time required for two adjacent units to complete a handshake protocol to transfer packages between them etc. Thus a small set of relevant
episodes can provide a good summary of the events in the conveyor system. We can use the discovered set of relevant episodes to achieve a lossless
compression of the original temporal event sequence to support remote monitoring, diagnostics and visualization activities. We explain the system in
more detail in Section~\ref{sec:conveyor}. Using data obtained from a high-fidelity discrete event simulator of such conveyor systems, we demonstrate
that our algorithms: (a) unearth a small set of relevant episodes that capture the essence of the transport through the system, and (b) our scheme
achieves good data compression. 

Even though our method is motivated  by the above application, we show that our method is effective with other general sequential data as well. Apart
from conveyor system data streams, we show the effectiveness of our methods with text data as well as on a few other real data sequences. These are
the data sets that are used to illustrate the effectiveness of the algorithms presented in~\cite{tatti2012long, lam2012mining,lam2014mining}. We
compare the performance of our algorithm with these methods on these data sets as well as on the composable conveyor system data.

The rest of the paper is organized as follows. In Section~\ref{sec:prelims}, we briefly review the formalism of episodes, introduce the new subclass
of serial episodes and formally state the problem. Section~\ref{sec:encoding-expln} describes our encoding scheme for temporal data using our
episodes. The various algorithms for mining and subset selection are explained in Section~\ref{sec:algo} and the experimental results are given in
Section~\ref{sec:expt}. We conclude the paper in Section~\ref{sec:conclude}.

\section{Problem Statement}
\label{sec:prelims}
 \subsection{Fixed Interval Serial Episodes}
The data we consider is a sequence of $n$ events denoted as $\scrD = \langle (E_1, t_1), (E_2,t_2), \newline
\ldots,(E_n, t_n) \rangle$, where $t_i \leq t_{i+1}$, and if $t_i = t_{i+1}$, then $E_i \neq E_{i+1}$, where
$E_i \in \Sigma$, is the event-type, $\Sigma$ is the alphabet and $t_i \in
\mathbb{Z}^+$ is the time stamp of the $i^{th}$ event. Note that we can have multiple events (of different types) all
occurring at the same time instant. A $k$-node serial episode $\alpha$ is denoted as $e_1\rightarrow e_2 \rightarrow
\cdots \rightarrow e_k$ where $e_i \in \Sigma$, $\forall i$.  An occurrence of $\alpha$ in $\scrD$ is a mapping
$h:\{1\ldots k\} \rightarrow \{1\ldots n\}$, such that $e_i = E_{h(i)}, 1 \leq i \leq k$ and $t_{h(i)} < t_{h(j)}$, for
$i < j$. An occurrence can be denoted by $(t_{h(1)}, \ldots, t_{h(k)})$, the event times of the events constituting the
occurrence. 
We call the interval  $[t_{h(1)}, t_{h(k)}]$ as the occurrence window of this occurrence. (If $k = 1$, then for the $1$-node episode, the occurrence 
window is essentially a number which is the event time of that event). 
Consider an example event sequence
\begin{multline}
\label{dataseq}
 \scrD_1 = \langle(A,1),(A,2),(B,3),(E,4),(A,5),(B,6),(C,6),(B,7),(D,8),\\ (C,10),(E,11)\rangle  
\end{multline}
In the data sequence given in~\eqref{dataseq}, a few occurrences of episode $A\rightarrow B \rightarrow C$ are  $(1,3,6), (2,3,6), (5,6,10)$.

A {\em fixed interval serial episode} is a serial episode with fixed inter-event gaps. A
fixed interval serial episode is denoted as  $\beta = e_1 \xrightarrow{\Delta_1} e_2 \xrightarrow{\Delta_2} \cdots
\xrightarrow{\Delta_{k-1}} e_k$. We will be considering the class of fixed interval serial episodes, where $\Delta_i
\leq T_g, \forall i$, with $T_g$ being a user specified upper bound on allowable gap. An occurrence of $\beta$ in
$\scrD$ is a mapping $h:\{1\ldots k\}\rightarrow \{1\ldots
n\}$, such that $e_i = E_{h(i)}, 1 \leq i \leq k$ and $t_{h(i+1)} - t_{h(i)} = \Delta_i>0$, for $1 \leq i < k$. For
example, in sequence $\scrD_1$ in \eqref{dataseq}, there are two occurrences of episode $A\xrightarrow{2} B
\xrightarrow{3} C$, namely $(1,3,6)$ and $(5,7,10)$. Note that the time of the first event of an occurrence completely specifies the
entire occurrence. This property of the fixed interval serial episodes allows us to design a coding scheme that results
in data compression. 
A $k$-node fixed interval serial episode $\alpha = e_1 \xrightarrow{\Delta_1} e_2 \xrightarrow{\Delta_2} \cdots
\xrightarrow{\Delta_{k-1}} e_k$ is called \mbox{{\it injective}} if
$e_i \neq e_j, \forall i,j, i \neq j$. 

In the literature, different notions of frequency are defined for episodes depending on the type of occurrences we
count. (For a discussion on various frequencies see \cite{achar2012unified}). An episode is said to be frequent if its 
frequency is above a given threshold. In this paper, we consider the number of distinct occurrences as the
frequency. Two occurrences are distinct if none of the events of one occurrence  is among events of the other. More
formally,  a set of occurrences, $\{h_1, h_2, \ldots, h_m\}$  of an episode $\alpha$ are \mbox{{\it distinct}} if  for
any $k \neq k'$, $h_k(i) \neq h_{k'}(j)$, $\forall i,j$. This is a natural notion of frequency for an injective fixed
interval serial episode because any pair of its occurrences with different start times will always be distinct. 

In this paper, we consider injective fixed interval serial episodes and from now on
we refer to injective fixed interval serial episodes simply as episodes whenever there is no scope for confusion.

\subsection{Selecting a Subset of Episodes Using the MDL Principle}

Pattern mining algorithms often output a large number of frequent episodes. Our goal is to isolate a 
small subset of them which are non-redundant and are relevant for the data. To formalize this goal, 
we use the MDL principle which views learning as data compression. The idea is that if we can discover all the relevant regularities in
the data, then an encoding based on these would result in data compression~\cite{grunwald2007minimum}. Thus, the goal
is to find a model which allows us to encode the data in a compact fashion.

Given any model, $H$, let $L(H)$ denote  the length for encoding the model $H$ and let $L(\scrD|H)$ be the length of the
data when encoded using the model $H$. Given an encoding scheme, under the MDL principle our goal is to find a model $H$
that minimizes total encoded length, $L(H, \scrD) = L(H) +L(\scrD|H)$.


For us,  different models correspond to different subsets of the set of frequent fixed interval serial episodes.
As mentioned earlier, an occurrence of such an episode is uniquely specified by its start time. Hence, by giving the
code for the identity of the episodes and a list of  start times, we can code all the events constituting the occurrences of this
episode. (We explain our encoding scheme in the next section). Thus, large episodes with many occurrences would account for a large number of events
in the data sequence thus
decreasing $L(\scrD|H)$. Another advantage of our use of the MDL principle is that it inherently takes care of redundancy. Selecting
episodes with minimal overlap among their occurrences would help reduce the final encoded length. 

Under MDL, we are looking at lossless coding and hence the occurrences of the selected subset of episodes have to {\em cover} the entire
dataset; i.e., every event in the data sequence should be part of an occurrence of (at least) one of the selected set
of episodes.  We can always ensure this by adding a few 1-node episodes, as needed. We will give details of our encoding
in the next section. Our main problem can now be stated as below
\begin{problem}
\label{prob1}
\sloppypar
Given a data sequence $\scrD$ and a set of (frequent) fixed interval serial episodes,  $\scrC = \{C_1, C_2, \ldots,
C_N\}$, find  the subset $H^* \subseteq \scrC$ such that  
\[H^* = \argmin_{H \subseteq \scrC}  \{L(H) + L(\scrD|H) \}\]
\end{problem}

\section{The Encoding Scheme for Data}
\label{sec:encoding-expln}
In this section, we explain our encoding scheme and derive the expression for encoded data length.

\subsection{Encoding}
Each model $H$  is a set of some fixed interval serial episodes whose occurrences cover the data. Given
such an $H$, which forms the dictionary, the data is then encoded by specifying the start times of selected occurrences
of the episodes. 
\begin{table}
\small
 \caption{A data sequence and its encoding}
 \label{table:encoding}
 $\scrD_2 = \langle(D,1)(A,2)(C,3)(E,3)(A,4)(B,4)(C,5)(D,5)\\(B,6)(C,7)(E,7)(C,8)(C,9)\rangle$\\\\
 \centering
 \begin{tabular}{cccc} \toprule
  Size of&	Episode &No. of & List of\\ 
   Episode & Name & Occurrences & Occurrences \\
  \midrule
  3 &$A\xrightarrow{2} B \xrightarrow{1} C$ &2 &$\langle 2, 4 \rangle$ \\
  3& $D\xrightarrow{2} E \xrightarrow{2}C$ &2 & $\langle 1, 5 \rangle$ \\
  1& $C$ & 2 &$\langle 3,8\rangle	$\\	\bottomrule
 \end{tabular}
 \end{table} 

 We explain our encoding scheme through an example. Table~\ref{table:encoding} shows a data sequence $\scrD_2$ and its encoding using three 
arbitrarily selected episodes. Each row in the table describes one of the episodes used and the encoding for the part of the data, covered by the 
occurrences of that episode. There are four columns in the table. Column 1 gives the size of the episode in that row and the second column specifies 
the episode. The third column gives the number of occurrences of that episode (used for encoding) and the last column gives a list of start times. 
Hence the first row of Table~\ref{table:encoding} specifies: a $3$-node episode, namely, $A\xrightarrow{2} B \xrightarrow{1} C$ and two of its 
occurrences starting at times 2 and 4. Thus, the first row of Table~\ref{table:encoding} accounts for the events $(A,2), (B,4), (C,5)$ and 
$(A,4), (B,6), (C,7)$ in the data, which are the events constituting the two occurrences of the episode $A\xrightarrow{2} B \xrightarrow{1} C$ 
starting at time stamps 2 and 4 respectively. We can think of the first two columns of the table as our dictionary and the last two 
columns of the table as the encoding of the data.  
Note that the seventh event in $\scrD_2$, $(C,5)$, is part of the occurrence of $A\xrightarrow{2} B \xrightarrow{1} C$
starting at 2 and of $D\xrightarrow{2} E \xrightarrow{2}C$ starting at 1. While this is allowed in our encoding,
minimizing such overlaps would improve total encoded length. In fact, we use injective fixed interval serial episodes in order to avoid overlap of 
different occurrences of the same episode, since, as we mentioned earlier, occurrences of injective fixed interval serial episodes starting at 
different time stamps are distinct.
occurrences 
In  Table~\ref{table:encoding}, the first two episodes account for all but two events in the data and hence we added a 1-node
episode (in row 3 of the table) to ensure that we cover the full data sequence. Our final encoded sequence would be a table like this. Each entry 
in the table is essentially a series of integers and our final encoded data would consists of a series of integers obtained by 
stringing together the rows in order. 

\subsection{Decoding}
In this section we discuss how to decode the encoded data. The encoded data consists of rows of a table, with each row  specifying an episode 
 and its occurrences. In each row, we read the first value, which is the size of the pattern. If this value is $k$, the next $2k-1$ integers
correspond to the codes of the event types ($k$ units)  followed by the inter event gaps ($k-1$ units). The next value in the encoded sequence
corresponds to the number of occurrences of the episode. We then need to read that many values to obtain all the occurrence start times of that
episode and complete reading the current row. Since we know when the row is complete, the next integer would be the first entry of the next row and we
repeat the same process as above. From each start time of occurrence of an episode, we can roll out the corresponding events because we know the event
types and the inter-event gaps. Once we are done with rolling out all the occurrences of all the episodes, we have to just sort the events based on
the time stamps and delete duplicate occurrences\footnote{Note that while events of different types can occur with the same time stamp, events with
same event-type cannot co-occur at the same time instant; hence we can easily spot duplicates while decoding.} to retrieve back the original data
sequence.

\subsection{Length of the Encoding}
\label{sec:enc-length}
We have seen that once the dictionary is fixed, the data encoding is  just a series of integers denoting the start times of occurrences of the
patterns in the dictionary. Even though, we could use bit level integer encoding schemes like Elias codes~\cite{lam2014mining, witten1999managing} and
Universal codes~\cite{tatti2012long, rissanen1983universal} for encoding integers (and have the size of encoding dependent on the value of the
integer), we use the notion of fixed memory units instead. The reason being that, the MDL principle looks at utilizing the regularity in data to
compress the data and hence the level of compression should not depend on the magnitude of data item. For example, the value of a time stamp, per se,
does not have any regularity and the compression achieved by the encoding scheme, hence, should not be dependent on the values of time stamps.
Therefore, for calculating the total encoded lengths we consider event types and times to be integers and assume that each such integer accounts for
one unit in the encoded length. Since our aim is to compare different models, keeping all lengths in terms of one unit per integer is sufficient for
us. (Here we are assuming that, in describing episodes in the dictionary, all event-types take the same amount of memory. We could, of course, use
codes such as Hamming codes, to reduce expected length of representation of episodes. We do not consider such extra compression here). Later on, while
comparing our method with the various other methods, we use bit level encoding for calculating lengths of integers so that we can easily compare with
the results of other algorithms.

Let model $H$ contain the episodes $\{F_1, F_2, \ldots, F_K \}$, with $|F_i|$ denoting the size of episode $F_i$. Each episode needs 
one integer to represent its size, 
$|F_i|$ integers for representing the event types and $(|F_i|-1)$ integers for inter-event gaps. Hence the
first two columns of the table need $\sum_{i=1}^K (1+|F_i|+|F_i|-1) = \sum_{i=1}^K 2|F_i|$ integers. This is $L(H)$. 

Let $f_i$ be the number of occurrences of $F_i$ listed in the column 4 of our table. Then, columns 3 and 4 together need
$\sum_{i=1}^K (f_i + 1)$ integers. This is $L(\scrD|H)$. Thus for the model $H$, the total encoded length is 
\begin{equation}
\label{eqn:encoding}
 L(H, \scrD) = L(H) + L(\scrD|H) = \sum_{i=1}^K 2|F_i|+\sum_{i=1}^K (f_i + 1)
\end{equation}
For the encoding given by Table~\ref{table:encoding}, the length for the first row is $1+(3+2)+1+2=9$ and it is easy 
to see that the total encoded length is $9+9+5 = 23$. 

The length of raw data can be taken to be $2|\scrD|$ where $|\scrD|$ is the number of events in the data. 
However, taking this as the length of uncompressed data may result in higher value for the compression achieved 
by an algorithm. This is because, even without finding any patterns, we can represent the raw data more compactly 
 by simply using only 1-node episodes in our encoding. If we have $M$ event types then we use $M$ $1$-node episodes for encoding. The total encoded 
length, using Equation~\eqref{eqn:encoding}, and taking $K=M$ would be $2M+|\scrD|+M$. (Note that $\sum_{i = i}^{M} f_i = |\scrD|$, because all 
occurrences of the $M$ 1-node episodes, together would exactly cover the data; also no event in the data would be part of occurrences of two
different 
$1$-node episodes). We call such an encoding {\it trivial encoding}. For the data sequence $\scrD_2$, the length for trivial encoding would be 
$5\times 2 +13 + 5 = 28$. Even though in this example the length of the trivial encoding is more than $2|\scrD|$, for real datasets we would have 
$|\scrD| \gg M$ and hence total encoded length of trivial encoding would be less than $2|\scrD|$. Hence in calculating data compression with our 
method, we compare the length of the trivial encoding with the length of the encoding using selected episodes. 

\section{Algorithms}
\label{sec:algo}
In this section, we consider algorithms for discovering a subset of episodes which achieves good compression. 
Finding the optimal subset of episodes to minimize total encoded length is known to be NP-Hard~\cite{lam2012mining, tatti2012long}. Hence the 
methods we present here are approximation algorithms to Problem~\ref{prob1}.
We begin by presenting algorithm CSC-1 ({\bf CSC} is for {\bf C}onstrained {\bf S}erial episode {\bf C}oding), which is a two
phase method. This consists of discovering all frequent episodes through a depth-first search algorithm followed by a
greedy method of selecting a subset based on maximum coverage and minimum overlap. We then present algorithm CSC-2, 
which directly mines for relevant fixed interval serial episodes from the data without first discovering all frequent
episodes. 

\subsection{First Algorithm: CSC-1}
We first explain our depth-first mining algorithm and the basis for the greedy strategy for subset selection before
describing the full CSC-1 algorithm (which is listed as Algorithm~\ref{algo:CSC}).

\subsubsection{Mining}
\label{sec:mining}

\begin{algorithm}[t]
  \caption{MineEpisodes($\scrD, T_g, f_{th}$)}
  \label{algo:MineEpisodes}
  \algorithmicrequire Sequence data $\scrD$, the maximum inter-event gap $T_g$ and frequency threshold, $f_{th}$\\
  \algorithmicensure The set of frequent episodes, $\scrC$
  
  \begin{algorithmic}[1]
    \State $\scrA \gets$ Set of all frequent $1$-node episodes in $\scrD$, along with occurrence lists.
    \ForAll {$A \in \scrA$}
      \State $\mbox{\bf{ExploreDFS}}(A, \scrA, T_g, f_{th})$ \Comment{All the frequent episodes } \label{call-explore}
      \Statex \Comment{ are added to the global list $\scrC$}
    \EndFor
  \end{algorithmic}
 \end{algorithm}

 \begin{algorithm}[t]
 \caption{ExploreDFS($\alpha, \scrA$,$T_g, f_{th}$)}
 \label{algo:ExploreDFS}
 \algorithmicrequire  Episode $\alpha$ with its occurrence list; $\scrA$: the set of frequent one node episodes with
its occurrence lists; $T_g$: Maximum inter-event gap; $f_{th}$: frequency threshold \\
 \algorithmicensure The set of frequent episodes $\scrC$.
 \begin{algorithmic}[1]
    \ForAll {$A \in \scrA \backslash$\{set of event-types in $\alpha$\}}
      \State $\mbox{{\em occurrlist-for-delta}} \gets \mbox{\bf{find-lists}}(\alpha, A, T_g)$	\label{findlist-call}
      \For {$j = 1 \to T_g$}
	\If{$|${\em occurrlist-for-delta}$(j)| \geq f_{th}\times |\scrD|$ } \label{thresh-check}
	  \State $\beta \gets (\alpha \xrightarrow{j} A)$
	  \State $\beta.occurrencelist \gets \mbox{{\em occurrlist-for-delta}}(j)$	
	  \State $\scrC \gets \scrC \cup \beta$
	  \State $\mbox{\bf{ExploreDFS}}(\beta, \scrA, T_g)$ \label{rec-explore-call}
	\EndIf
      \EndFor
    \EndFor
  \end{algorithmic}
 \end{algorithm}

\begin{algorithm}[t]
 \caption{find-lists$(\alpha, A, T_g)$}
 \label{algo:find-lists}
 \algorithmicrequire Episode $\alpha$ and one node episode $A$ with their occurrence lists.\\
 \algorithmicensure The array, {\em occurrlist-for-delta} storing the occurrence lists with different gaps.
 \begin{algorithmic}[1]
  \ForAll {$[t_s^{\alpha}, t_e^{\alpha}] \in \alpha.occurrencelist$}
    \State Let $t^A$ be the first occurrence of $A$ after $t_e^{\alpha}$ \Comment{NULL if no such occurrence}

    \While{$t^A \neq NULL$ and $t^A - t_e^{\alpha} \leq T_g$} \label{findlists-while}
      \State $j \gets t^A - t_e^{\alpha}$; 
      \State Add $[t_s^{\alpha}, t^A]$ to {\em occurrlist-for-delta}$(j)$ \Comment{Corresponding to $\alpha \xrightarrow{j} A$}
\label{findlists-addwind}
      \State $t^A \gets$ next occurrence of $A$	 \Comment{NULL if there is no next occurrence}
    \EndWhile
  \EndFor
  \State {\bf{return}} {\em occurrlist-for-delta}
 \end{algorithmic}
\end{algorithm}

To obtain all  frequent injective fixed interval serial episodes, we use a depth-first (also known as pattern-growth) strategy using occurrence 
windows. See ~\cite{avinash2013pattern, MR04} for more details on depth-first strategies using occurrence windows. 
The basic idea is as follows. First we find all 1-node frequent episodes (which are event-types that occur often enough) and for each frequent 1-node 
episode, keep its occurrence list which is a list of event times where the 1-node episode occurs in the data. Let $\alpha$ be an 
episode and suppose we are given  a list  of all its occurrence windows (also called occurrence list). Recall that the occurrence window of an
episode $\alpha$ is an interval $[t_s, t_e]$, where $t_s$ and $t_e$ are the times of the first and last events of $\alpha$ in this occurrence. If we
know all occurrence windows of 1-node episode $A$, then we can easily check whether $[t_s, t_e+j]$ is an occurrence window of $\alpha\xrightarrow{j}
A$. Thus we can easily calculate the occurrence windows of episodes such as $\alpha \xrightarrow{j} A$, for all allowed $(A,j)$, and hence find
frequent episodes of next size. By doing this recursively, we find all frequent episodes. The implementation of this idea is described in
Algorithms~\ref{algo:MineEpisodes} to \ref{algo:find-lists}.

The main function is {\it MineEpisodes}, listed as Algorithm~\ref{algo:MineEpisodes}. This is a wrapper function, which finds, for each event-type $A 
\in \scrA$, where $\scrA$ is the set of frequent 1-node episodes, all the {\em frequent} fixed interval serial episodes with $A$ as the prefix, using 
the {\it{ExploreDFS}} function (in line~\ref{call-explore}). The frequency threshold value, $f_{th}$ is user specified, and is given as a fraction
of the data length (see line~\ref{thresh-check}). 

The procedure {\it {ExploreDFS}}, listed as Algorithm~\ref{algo:ExploreDFS}, is a recursive function. Given an  input episode $\alpha$, it finds
all the frequent right extensions of $\alpha$, i.e., all the frequent episodes with $\alpha$ as prefix. For each $A$, the procedure initially finds
the occurrence windows for the episodes $\alpha \xrightarrow{j} A, 1 \leq j \leq T_g$, where $T_g$ is the maximum allowed inter-event gap
(line~\ref{findlist-call}, Algorithm~\ref{algo:ExploreDFS}), by calling the procedure {\it find-lists}, which is explained below.  
The function then recursively goes deeper for each frequent episode, $\alpha \xrightarrow{j} A$ (line~\ref{rec-explore-call}). 

The {\it{find-lists}} procedure (listed as Algorithm~\ref{algo:find-lists}), takes as input, episode $\alpha$ and event type $A$ (which is also a 
1-node episode), and finds the occurrence windows for all the episodes $\alpha \xrightarrow{j} A, j \leq T_g$.
For each occurrence window $[t_s^{\alpha}, t_e^{\alpha}]$ of $\alpha$, it looks
for all the occurrences $t^A$ of $A$ such that the new occurrence window satisfies the maximum inter-event gap
constraint $T_g$ (the condition for while loop in line~\ref{findlists-while}). An occurrence window satisfying the constrain is then added to the
occurrence list corresponding to the episode $\alpha \xrightarrow{j} A$, where $j = t^A - t_e^{\alpha}$ (line~\ref{findlists-addwind}).

Using these algorithms, we get all injective frequent fixed interval serial episodes. We then go on to select the best representative subset. 

\subsubsection{Selection Strategy}
Given a data sequence $\scrD$, let $\alpha$ be an $N$-node fixed interval serial episode with frequency
$f^{\alpha}_{\scrD}$. We define the score of $\alpha$ in $\scrD$ as
\begin{equation}
\label{eqn:score}
  score(\alpha, \scrD) = f^{\alpha}_{\scrD}N - (2N + f^{\alpha}_{\scrD} + 1)
\end{equation}
Recall that $2N + f^{\alpha}_{\scrD} + 1$ is the total encoded length for encoding all the events that constitute the $f^{\alpha}_{\scrD}$ 
occurrences of $\alpha$. (Recall from Section~\ref{sec:enc-length} that $2N$ is the  length for encoding $\alpha$ and $f^{\alpha}_{\scrD} + 1$ is the
total length for encoding all $f^{\alpha}_{\scrD}$ occurrences of $\alpha$.). 
It is easy to see that $f^{\alpha}_{\scrD}N$ is a lower bound on the encoded data length for trivially encoding all the events in the 
$f^{\alpha}_{\scrD}$ occurrences of (the $N$-node episode) $\alpha$ with 1-node episodes\footnote{We note that this is a lower bound because 
$\alpha$ is an injective episode. When $\alpha$ is an injective episode, no two occurrences of $\alpha$ can share an event and hence
$f_\scrD^\alpha$ occurrences would contain $f_\scrD^\alpha N$ events in the data sequence. If the episodes were non-injective, then there is a
possibility of events being shared by different occurrences of the same episode and hence $f^{\alpha}_{\scrD}N$ would not be the lower bound.}. 

If $score(\alpha, \scrD) > 0$, then $\alpha$ is called a {\it useful candidate} since selecting it can improve encoding
length by at least the value of $score(\alpha, \scrD)$, in comparison to trivial encoding. From Equation~\eqref{eqn:score}, we can easily see that
$score(\alpha, \scrD) > 0$, if $f^{\alpha}_{\scrD} > \frac{2N+1}{N-1}$. 
Thus the episode $\alpha$ will be a useful candidate if $f^{\alpha}_{\scrD} > 5$ for $|\alpha|=2$ and $f^{\alpha}_{\scrD} > 3$ for $|\alpha| \geq 3$.
But selecting any useful candidates would not lead to efficient encoding. For any pair of selected episodes, we also want their occurrences to have
least number of events in common.
Our subset selection procedure for encoding the data is based on greedy selection of episodes whose occurrences cover
large number of events in the data and have low level of overlap with the occurrences of other selected episodes. 

Let $\scrF_s = \{\beta_1, \beta_2, \ldots, \beta_S\}$ be a set of episodes of size greater than one.  Given any such $\scrF_s$, let
Let $L(\scrF_s, \scrD)$ denote the total encoded length of $\scrD$, when we encode all the events which
are part of the occurrences of episodes in $\scrF_s$, by using episodes in $\scrF_s$ as per our encoding scheme and encode the remaining events in 
data, if any,  by episodes of size one. 

Given any two episodes $\alpha,\beta$, let $OM(\alpha,\beta)$ denote the number of events in the data that are covered by occurrences of {\em both} 
$\alpha$ and $\beta$. We call $OM$ the {\it Overlap Matrix}. 

We define, for $\alpha \notin \scrF_s$
\begin{equation}
\label{eqn-overlap}
 overlap\mbox{-}score(\alpha, \scrD, \scrF_s) =  f^{\alpha}_{\scrD}N - \sum_{\beta_i\in \scrF_s} OM(\alpha,\beta_i)- (2N + f^{\alpha}_{\scrD} + 1) 
\end{equation}
Note that $overlap\mbox{-}score(\alpha, \scrD, \scrF_s) = score(\alpha, \scrD) - \sum_{\beta_i\in \scrF_s} OM(\alpha,\beta_i)$ and is another
measure for the gain
in encoding length, when we add $\alpha$ to $\scrF_s$. The measure has an interesting property as explained below.
\begin{proposition}
  \label{prop:overlap}
 If $overlap\mbox{-}score(\alpha, \scrD, \scrF_s) > 0$, then $L(\scrF_s, \scrD) > L(\scrF_s \cup \{\alpha \}, \scrD)$
\end{proposition}
\begin{proof}
First, note that the difference in encoding will only be in the section of the data, where the encodings 
using $ \scrF_s$ and $\scrF_s \cup \{\alpha \}$ differ. As is easy to see, $\sum_{\beta_i\in \scrF_s} OM(\alpha,\beta_i)$ is an upper bound on the
number of
events of the occurrences of $\alpha$, that are shared with the occurrences of episodes in $\scrF_s$. 
Hence, $f^{\alpha}_{\scrD}N  - \sum_{\beta_i\in \scrF_s} OM(\alpha,\beta_i)$ is a lower bound on the number of events not covered by anyone in
$\scrF_s$ and which are covered by the occurrences of $\alpha$. Hence if we use $\scrF_s$, it takes at least $f^{\alpha}_{\scrD}N -
\sum_{\beta_i\in \scrF_s} OM(\alpha,\beta_i)$ units for encoding these events using size-1 episodes. In contrast, by adding $\alpha$
to the $\scrF_s$, it takes $2N + f^{\alpha}_{\scrD} + 1$ units to encode those occurrences, independent of the number
of events $\alpha$ shares with other episodes. Thus the reduction in encoding length between the two  is at least
$f^{\alpha}_{\scrD}N - \sum_{\beta_i\in \scrF_s} OM(\alpha,\beta_i) - (2N + f^{\alpha}_{\scrD} + 1)$, which is the
$overlap\mbox{-}score(\alpha, \scrD, \scrF_s)$. 
 Hence the result. 
\end{proof}
Proposition~\ref{prop:overlap} says that, with respect to an already selected set of episodes $\scrF_s$, adding an
episode $\alpha$ with $overlap\mbox{-}score(\alpha, \scrD, \scrF_s) > 0$ to the set $\scrF_s$, would only reduce the total length of encoded data.
Our greedy heuristic is to select the episode with the maximum overlap-score. (Note that, by definition, $overlap\mbox{-}score(\alpha, \scrD, \scrF_s)
= score(\alpha, \scrD)$, if $\scrF_s = \emptyset$). 

\subsubsection{CSC-1}
The CSC-1 algorithm selects a best subset of the frequent episodes based on minimizing encoding length. The algorithm takes as input, $K$, the maximum
number of episodes (of size greater than 1) in the final selected subset. Thus, it can be used a method to select the `best-$K$' episodes or to select
the best subset to achieve maximum compression (by choosing a very large value of $K$). 

The CSC-1 algorithm for selecting a good subset of (maximum $K$) episodes is listed as Algorithm~\ref{algo:CSC}. The algorithm runs in iterations of 
mining frequent fixed interval serial episodes from the data sequence, then selecting, one by one, a set of good encoding episodes from the mined set 
and finally deleting the occurrences of the selected episodes from the sequence. The process is repeated until we found $K$ good episodes or we
cannot find any episode that can give any gain in encoding. 

In each iteration of the {\it while} loop (lines~\ref{CSC1-while}-\ref{CSC1-endwhile}), we first mine the set of frequent fixed
interval serial episodes, $\scrC$, using the {\it MineEpisodes} procedure, explained in Section~\ref{sec:mining} (line~\ref{cand set}).
We next calculate the $OM$ matrix (line~\ref{FindOverlapCall}) and then calculate the $overlap\mbox{-}score$. 
Each iteration of the {\it repeat} loop (lines~\ref{repeat}-\ref{until}) looks for an episode in the current candidates set, $\scrC$, which has the
highest positive $overlap\mbox{-}score$. If such an episode is found, it is added to the set $\scrF_s$. This greedy strategy is justified by
Proposition~\ref{prop:overlap}. The set $\scrF_s$, thus contains all the episodes selected in the repeat loop. The repeat loop is broken, when no
episode with positive $overlap\mbox{-}score$ exists in the current candidate set or we have selected $K$ episodes(line~\ref{until}). Then all the 
events in the occurrences of the episodes in $\scrF_s$ are deleted from the data (line~\ref{removeevents}). We then once again repeat the process of 
finding frequent episodes from this modified data and selecting a subset of episodes from this episode set.

The while loop runs as long as the selected set size is less than $K$ and it finds at least one episode that increases encoding efficiency. This 
condition is checked in lines~\ref{breakingcondition}-\ref{CSC1-coveringfalse}. When we cannot find any more episodes with positive 
$overlap\mbox{-}score$, we encode the remaining events in the data with 1-node episodes 
(lines~\ref{CSC1-RemEventsEncoding1}-\ref{CSC1-RemEventsEncoding2}).

\begin{algorithm}[t]
 \caption{CSC-1($\scrD, T_g, f_{th}, K$)}
 \label{algo:CSC}
 
 \algorithmicrequire Data sequence $\scrD$; maximum inter-event gap $T_g$; threshold $f_{th}$;
 maximum number of selected episodes $K$.\\
 \algorithmicensure The set of selected frequent episodes $\scrF$
 \begin{algorithmic}[1]
  \State $\scrF \gets \emptyset$
  \State $coveringexists \gets true$
  \While{$coveringexists$ {\bf and} $|\scrF| < K$}  \label{CSC1-while}
    \State  $\scrF_s \gets \emptyset$
    \State $\scrC \gets$ {\bf MineEpisodes}($\scrD, T_g, f_{th}$) \label{cand set}
    \State $OM \gets$ {\bf FindOverlapMatrix}($\scrD, \scrC$)  \label{FindOverlapCall}
    \Repeat \label{repeat}
      \State $\alpha \gets \argmax_{\gamma \in \scrC} overlap\mbox{-}score(\gamma, \scrD, \scrF_s) $
\label{findbest}
      \If {$overlap\mbox{-}score(\alpha, \scrD, \scrF_s) > 0$}
	\State $\scrF_s \gets \scrF_s \cup \{\alpha\}$ 
	\State $\scrC \gets \scrC \backslash \alpha$
      \ElsIf {$overlap\mbox{-}score(\alpha, \scrD, \scrF_s) \leq 0$ and $\scrF_s = \emptyset$} \label{breakingcondition}
	  \State $coveringexists \gets false$  \label{CSC1-coveringfalse}

      \EndIf
    \Until{$overlap\mbox{-}score(\alpha, \scrD, \scrF_s) \leq 0$ {\bf or } $|\scrF \cup \scrF_s| = K$} \label{until}
    
    \State $\scrD \gets \scrD \backslash (occurrences\ of\ \scrF_s)$ \label{removeevents}
    \State $\scrF \gets \scrF \cup \scrF_s$
  \EndWhile 	\label{CSC1-endwhile}
  \State $\scrA \gets$ Size-1 episodes in remaining $\scrD$ \label{CSC1-RemEventsEncoding1}
  \State $\scrF \gets \scrF \cup \scrA$ \label{CSC1-RemEventsEncoding2}
  \State \Return $\scrF$
 \end{algorithmic}
 \end{algorithm}

  The only remaining part in the CSC-1 algorithm is the calculation of the matrix $OM$ (line~\ref{FindOverlapCall}), which we explain now. The
procedure {\it FindOverlapMatrix}, listed as Algorithm~\ref{algo:FindOverlap}, utilizes the occurrence lists for all
the frequent episodes (obtained from Algorithms~\ref{algo:ExploreDFS} and~\ref{algo:find-lists}), to calculate $OM$ matrix by one more pass over data
using the standard Finite State Automata(FSA) based method for tracking episode
occurrences~\cite{mannila1997discovery,laxman2007fast,achar2012unified}.
FSAs in \cite{mannila1997discovery,laxman2007fast,achar2012unified} are used to track occurrences of episodes. Algorithm~\ref{algo:FindOverlap}
uses FSAs for a different purpose since we already have the occurrences of the episodes. 
Here, the FSAs associated with episodes, and hence some times called episode automata, are used to find, for each event $(E_i,t_i)$ in the sequence,
the set of episodes for which this event is part of one of their occurrences. For all such pairs of episodes in that set, the $OM$ matrix is
incremented by 1. 

In Algorithm~\ref{algo:FindOverlap}, each state of an episode automaton specifies the state of a current occurrence and is denoted by $(\alpha, j,
t_s)$, where $\alpha$ is the episode, $j$ is the state of the automaton to which it is expecting to transit (this means that for the current
occurrence of the episode, it has seen events for $\alpha[1]$ to $\alpha[j-1]$ satisfying the inter-event constraints and is waiting for the event to
occur with event-type $\alpha[j]$\footnote{$\alpha[i]$ denotes the $i$th episode event of $\alpha$.}), and $t_s$ denotes the start time of the current
occurrence. An episode automaton corresponding to a $N$ sized episode has $(N+1)$ states. An automaton for episode $\alpha$ is in state $j = 1$
while it is waiting for the event with event-type $\alpha[1]$ and is in state $j = |\alpha|+1$ at the end of its occurrence.

Each event-type, $E$, is associated with a data structure called $waits$. $waits(E)$ contains the list of automaton waiting for the event with
event-type $E$ to occur. Whenever an automaton corresponding to an episode $\alpha$ transits to a state $j$, in wait for an event-type $E$, that
corresponding automaton state $(\alpha, j, t_s)$ is added to $waits(E)$ (lines~\ref{waits-add1} and \ref{waits-add2},
Algorithm~\ref{algo:FindOverlap}). 

For each episode, $\alpha$, as we go along the sequence, the start of one of its occurrences (which we know
apriori) initiates an automaton for the episode $\alpha$. As we parse through the events in $\scrD$ ({\bf for} loop in line~\ref{forloop-event}), each
automaton corresponding to occurrences of the episodes will be waiting for a specific event to come up. 
On seeing the event $(E_i,t_i)$, those episodes (automaton) that were waiting for the event type, $E_i$, (which is obtained from  $waits(E_i)$ (the
inner {\bf for} loop)) compares the start times with $t_i$ to see whether the event is part of the current occurrence (line~\ref{startcheck} and
line~\ref{constraintcheck}, Algorithm~\ref{algo:FindOverlap}). If the constraints are satisfied, then the episode is noted as having accepted that
event (in line~\ref{acceptance1} and line~\ref{acceptance2}, Algorithm~\ref{algo:FindOverlap}) and moves to the next state. Finally for those episodes
that have accepted the current event, the corresponding $OM$ counts are incremented (in line~\ref{OM-update}). 

\begin{algorithm}[t]
  \caption{FindOverlapMatrix($\scrD, \scrC$)}
  \label{algo:FindOverlap}
  \algorithmicrequire $\scrD$, the sequence dataset; $\scrC$, the set of frequent multiple node episodes in $\scrD$.\\
  \algorithmicensure $OverlapMatrix$, the $|\scrC| \times |\scrC|$ matrix denoting the number of shared events.
  \begin{algorithmic}[1]
    \ForAll{$\alpha \in \scrC$}
      \State Add automate state $(\alpha, 1, \emptyset)$ to $waits(\alpha[1])$
    \EndFor
    \ForAll{event $(E_i, t_i) \in \scrD$} \label{forloop-event}
      \State $overlaplist \gets \emptyset$
      \ForAll{automata $(\alpha, j, t_s) \in waits(E_i)$}
	\If{$j = 1$}
  	  \If{$\exists$ occurrence of $\alpha$ starting at $t_i$} \label{startcheck}
	    \State Add $\alpha$ to $overlaplist$	  \label{acceptance1}
	    \State Add automata state $(\alpha, j+1, t_i)$ to $waits(\alpha[2])$ \label{waits-add1}
  	  \EndIf
	\Else
	  \If{$(t_i - t_s) = \alpha.\Delta[j-1]$}  \label{constraintcheck}
	    \State Add $\alpha$ to $overlaplist$	  \label{acceptance2}
	    \If{$j \neq |\alpha|$} \Comment{If $j = |\alpha|$, we just retire the automaton}
	      \State Add$(\alpha, j+1, t_s)$ to $waits(\alpha[j+1])$ \label{waits-add2}
	    \EndIf
	    \State Remove $(\alpha, j, t_s)$ from $waits(E_i)$
	  \EndIf
	\EndIf
      \EndFor
      \State Increment $OverlapMatrix[\alpha][\beta]$ by 1, $\forall \alpha, \beta \in overlaplist, \alpha \neq \beta$  \label{OM-update}
    \EndFor
  \end{algorithmic}

 \end{algorithm}
\subsubsection{Complexity of the Algorithm}
There are three main steps in each iteration of Algorithm~\ref{algo:CSC}: the mining step (call to MineEpisodes in
line~\ref{cand set}), the calculation of the
$OM$ matrix (line~\ref{FindOverlapCall}) and the selection step (the repeat loop). 
The mining step involves the generation of occurrence start times of the selected episodes. For each episode $\alpha$, this is done by a single pass
over the occurrences of its prefix subepisode of size $(|\alpha|-1)$ and 1-node suffix subepisode. The number of occurrences is of the order
of $|\scrD|$ and hence for $\scrC$
selected episodes, the mining step takes $O(|\scrC||\scrD|)$.

The $OM$ matrix calculation depends on the number of episodes accepting each event in $\scrD$. Let $K_i$ denote the number of episodes
in $\scrC$ containing event-type $A_i$. Let $P_i$ denote the fraction of the events of type $A_i$ in $\scrD$. Then each event-type $A_i$ occurs
$P_i|\scrD|$ times and each of its occurrence is associated with $O(K_i^2)$ updates of the $OM$ matrix. Hence, the runtime for the calculation  of the
$OM$ matrix is $O((\Sigma_{i=1}^M P_iK_i^2)|\scrD|) = O(|\scrC|^2|\scrD|)$, since $K_i < |\scrC|$. 

The repeat loop involves calculating the $overlap\mbox{-}score(\gamma, \scrD, \scrF_s), \forall \gamma \in \scrC$, which takes $O(|\scrC||\scrF_s|)$
and hence for $|\scrF_s|$ iterations takes $O(|\scrC||\scrF_s|^2)$. Since $|\scrF_s| < |\scrC|$ the total computation for each iteration of our method
is $O(|\scrC|^2|\scrD| + |\scrC|^3)$. Hence the total execution is critically dependent on $|\scrC|$, the number of frequent episodes mined from the
data.

\subsection{An improved algorithm: CSC-2}

The runtime of the CSC-1 method as shown in the previous section is $O(|\scrC|^2|\scrD| + |\scrC|^3)$ and hence depends directly on the size of the
mined candidate set of episodes, $|\scrC|$. Thus an increase in the size of $\scrC$ would increase the runtime by a cubic order. CSC-2 addresses
exactly this problem. CSC-2 proceeds exactly as CSC-1 except that the set $\scrC$ (line~\ref{cand set}, Algorithm~\ref{algo:CSC}) that it considers at
each iteration is a small set rather than the full set of frequent fixed interval serial episodes. Earlier, in the {\it MineEpisodes} function, we
have mined for {\it all} frequent episodes in a depth-first manner and gave that as the set $\scrC$. Now, instead of that, we obtain a set of
episodes, each one of which is a best possible episode in one of the paths in the depth-first search tree. Here, the best episode is decided in a
greedy fashion based on its contribution to coding efficiency. Note that the CSC-2 algorithm does not need any frequency threshold to be specified by
the user because we do not mine for `frequent' episodes.

The algorithm CSC-2 is same as Algorithm~\ref{algo:CSC} except that the call to {\it MineEpisodes} (line~\ref{cand set}, Algorithm~\ref{algo:CSC}) is
replaced by a call to a new procedure {\it BestExtensions}. The rest of Algorithm~\ref{algo:CSC} remains the same for both CSC-1 and CSC-2. Hence, we
do not provide separate pseudo code for CSC-2. The procedure {\it BestExtensions} , which selects a set of candidate episodes, $\scrC$, is listed as
Algorithm~\ref{algo:BestExtensions}.
\begin{algorithm}[t]
 \caption{BestExtensions($\scrD$,$T_g$)}
 \label{algo:BestExtensions}
 \algorithmicrequire Data sequence $\scrD$; maximum inter-event gap $T_g$.\\
 \algorithmicensure Selected set of relevant candidate episodes $\scrC$.
 \begin{algorithmic}[1]
  \State $\scrA \gets \mbox{Set of all 1-node episodes in $\scrD$} $
    \ForAll {$A \in \scrA$}
      \State $patt \gets A$
      \State $newpatt \gets patt$
      \While{$(patt \gets Extensions(patt, \scrA, T_g))) \neq NULL$}
	\State $newpatt \gets patt$
      \EndWhile
      \If {$score(newpatt, \scrD) > 0$}
	\State $\scrC \gets \scrC \cup \{newpatt\}$
      \EndIf
    \EndFor
    \State \Return $\scrC$
  \end{algorithmic}
 \end{algorithm}

\begin{algorithm}[t]
 \caption{Extensions$(\alpha, \scrA, T_g)$}
 \label{algo:Extensions}
 \algorithmicrequire  Episode $\alpha$ with its occurrence list; $\scrA$: the set of frequent one node episodes with
its occurrence lists; $T_g$: Maximum inter-event gap.\\
 \algorithmicensure Episode extension with the highest frequency and minimum inter event gap.
 \begin{algorithmic}[1]
    \State $maxfreq \gets \frac{2*(|\alpha|+1)+1}{|\alpha|}$ \Comment{Min frequency for which coding efficiency is}
    \Statex \Comment{achieved}
    \State $mingap \gets T_g$
    \State $maxpatt \gets NULL$
    \ForAll {$A \in \scrA \backslash$\{set of event-types in $\alpha$\}}
      \State $\mbox{{\em occurrlist-for-delta}} \gets \mbox{\bf{find-lists}}(\alpha, A, T_g)$ \Comment{find-lists in Algorithm~\ref{algo:find-lists}}
      \State $bestgap \gets \argmax_{j\leq T_g} |\mbox{{\em occurrlist-for-delta}}(j)|$
      \If {$|\mbox{{\em occurrlist-for-delta}}(bestgap)| \geq maxfreq$} \label{extendcond:1}
	\If {$|\mbox{{\em occurrlist-for-delta}}(bestgap)| > maxfreq$ OR $bestgap < mingap$} \label{extendcond:2}

 	  \State $\beta \gets (\alpha \xrightarrow{bestgap} A)$
	  \State $\beta.occurrlist \gets \mbox{{\em occurrlist-for-delta}}(bestgap)$	
	  \State $\beta.freq \gets |\mbox{{\em occurrlist-for-delta}}(bestgap)|$
	  \State $maxpatt \gets \beta$
	  \State $maxfreq \gets \beta.freq$
	  \State $mingap \gets bestgap$

	\EndIf
      \EndIf
    \EndFor
    \State \Return $maxpatt$
  \end{algorithmic}
\end{algorithm}

For each event-type $A$, Algorithm~\ref{algo:BestExtensions} extends the 1-node episode $A$ with an event-type $B$
and gap $i$ to form the episode $A \xrightarrow{i}B$ such that 
\[ <B, i> = \argmax_{\substack{<C,j>,\\ C\neq A, j\leq T_g}} score(A \xrightarrow{j}C, \scrD)\] 
The episodes are grown by right extension until none of the immediate extensions has a better score than the current
episode. At this point we add the episode to the list of candidate episodes $\scrC$. 

The extension algorithm is given in Algorithm~\ref{algo:Extensions}. Among the immediate extensions,
Algorithm~\ref{algo:Extensions} selects the one with the maximum frequency. If multiple
extensions have the same frequency, then the one with the minimum gap is selected (lines~\ref{extendcond:1}
and~\ref{extendcond:2}, Algorithm~\ref{algo:Extensions}).

The number of candidates generated using this procedure is at most $M$, the size of the alphabet. In general, this is
much smaller than the size of $\scrC$ in CSC-1. However, for large alphabet (i.e., large number of event types), even
this could be costly. Then we can modify CSC-2 to keep only the `best few' episodes in $\scrC$.

\section{Experiments}
\label{sec:expt}
In this section, we present the experimental results for our method and compare its performance with that of 
SQS~\cite{tatti2012long} and GoKrimp~\cite{lam2014mining}. Since we observed that algorithm CSC-2 gives similar patterns
as CSC-1, but is much more efficient time-wise, we present simulation results with CSC-2 only. 

We consider three different types of data. The first type of data is the conveyor system data, since this application was the motivation for us
to come up with the new subclass of fixed interval serial episodes and our novel encoding scheme. The set of data sequences that we consider are
generated by a detailed simulator of composable conveyor systems. As we mentioned earlier, this is an application area where compression achieved may
be useful on its own. We explain more about the problem and the data in Section~\ref{sec:conveyor}.

The second type of data that we consider is text data. We consider JMLR data set which contains 787 abstracts from the Journal of Machine Learning
Research. This is a sequential symbolic data, but the time stamps, so to say, are just the serial number of the word in the sequence. This dataset is
used to see how the various algorithms perform in unearthing relevant phrases (patterns) related to machine learning research. We show the top 20
patterns found by different methods on this dataset. 

The final collection of data that we consider consists of five real-world data sets introduced in~\cite{moerchen2010robust}. Each of these is
a database of symbolic interval sequences with class labels; i.e., events in these databases are denoted by a symbol and an interval
of its occurrence. As in~\cite{lam2014mining}, we consider the start and end of the interval to be two different events with different event-types.
For example, the event interval $(e,t1,t2)$, where $e$ is the symbol and $[t1,t2]$ is the interval of its occurrence, would be considered as two
events $(e-, t1)$ and $(e+,t2)$. Table~\ref{table:datasummary} gives some relevant details regarding these data sets.
The reason for using these datasets is that these are the only real world data sets on which
results of performance by SQS and GoKrimp are reported. On all these data sets we compare the data compression and the
classification accuracy of our method with that achieved by SQS and GoKrimp.


\begin{table}
  \centering
 \caption{Summary of Datasets.}
 \label{table:datasummary}
  \begin{tabular}{@{}ccccc@{}} \toprule
  Datasets 	& Events & Sequences & Classes &Alphabet Size, $M$  \\ \midrule
   jmlr		&75646& 787 &NA& 3846\\ 
   aslbu		&36500& 441 &7 & 190\\
   aslgt		&178494& 3493 &40& 47\\
   auslan2	&1800&200 & 10& 16\\
   context	&25832 & 240 & 5& 56\\
   pioneer	&9766 & 160 & 3& 92\\
   \bottomrule
  \end{tabular}
\end{table}

Our algorithms are implemented in C++ and the experiments were executed single threaded on an Intel i7 4-core processor
with 16 GB of memory running over a linux OS. The source code for our algorithms will be made available on request. The
implementations of SQS and GoKrimp were obtained from the respective authors. For the GoKrimp algorithm, there is a
significance level parameter (for a test called sign test), which is set to the default value of 0.01 and the minimum
number of pairs needed to perform a sign test is set to the their default value of 25. For our algorithm, the maximum
inter-event gap was set to 5 for all the datasets. The $K$ value in the CSC-2 algorithm, denoting the 
maximum number of patterns to be selected is set to infinity (which means maximum possible selection of episodes) unless otherwise stated.

\subsection{Results on Composable Conveyor System data sequences}
Before giving the results, we first give a brief explanation about composable conveyor systems~\cite{sks2006,asrr2009}.
\subsubsection{Conveyor System}
\label{sec:conveyor}
Material handling conveyor systems move packages or material units from one or more inputs to specified outputs along predetermined paths. Such
systems are used extensively in manufacturing, packaging, packet sorting etc. A typical conveyor system with two inputs and two outputs is shown
in Fig.~\ref{fig:2I-2O-topo}. This system is composed using instances of Segment and Turn units that each operate autonomously. A Segment moves a
package over a predetermined length over its belt. A Turn is a unit that can serve as a merger or splitter for package flow. Each segment and turn
unit has a predetermined maximum speed of operation and each unit has local sensors and actuators that are used to autonomously control its local
behavior. The simulator we used produces a detailed event-trace of every change of state (event) in the sensors and actuators of each unit as packages
move through the system.

The movement of the packages in the conveyor system is transformed to a temporal 
sequence as follows. We label the exit points of each input, output, turn and segment.
We consider the transfer of a package from one conveyor unit to another as an event. The unit represents the type of the
event and the time at which each event occurs is recorded as the event time.
Hence, when a package moves along a path, an ordered sequence of events is recorded. All the events corresponding to one
package may not be contiguous because during the same time other packages may be moving along other paths. 


We consider three conveyor system datasets coming from three different topologies. Each topology can
be identified by the various paths from the inputs to the outputs and the package input rates. We consider the
package input arrival rates to be Poisson. The details of the three datasets is given in Table~\ref{tableTopos} and
the structure of the topologies are given in Figures~\ref{fig:2I-2O-topo}, \ref{fig:3I-3O-topo} and
\ref{fig:Packet-Sorter-topo}.

\begin{figure}
  \centering
  \includegraphics[width=0.6\textwidth,height=4.cm]{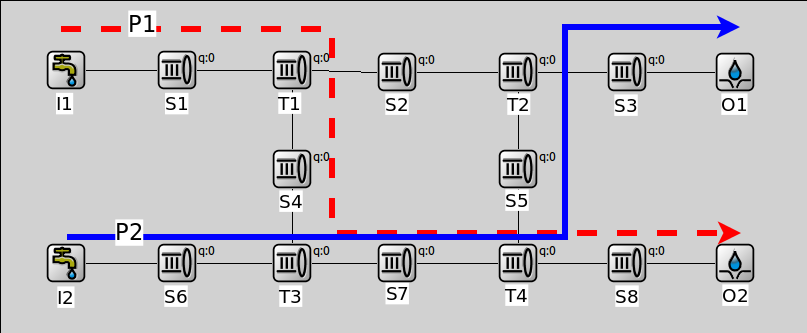}
   \caption{A two-input, two-output conveyor system. Packages enter via 
   one of the inputs $I1$ or $I2$ and leave the system via one of the outputs as illustrated.
   The system is composed using instances of Segments and Turns~\cite{asrr2009,sks2006}.}
  \label{fig:2I-2O-topo}
 \end{figure}

\begin{figure}
  \centering
  \includegraphics[width=0.7\textwidth,height=5.cm]{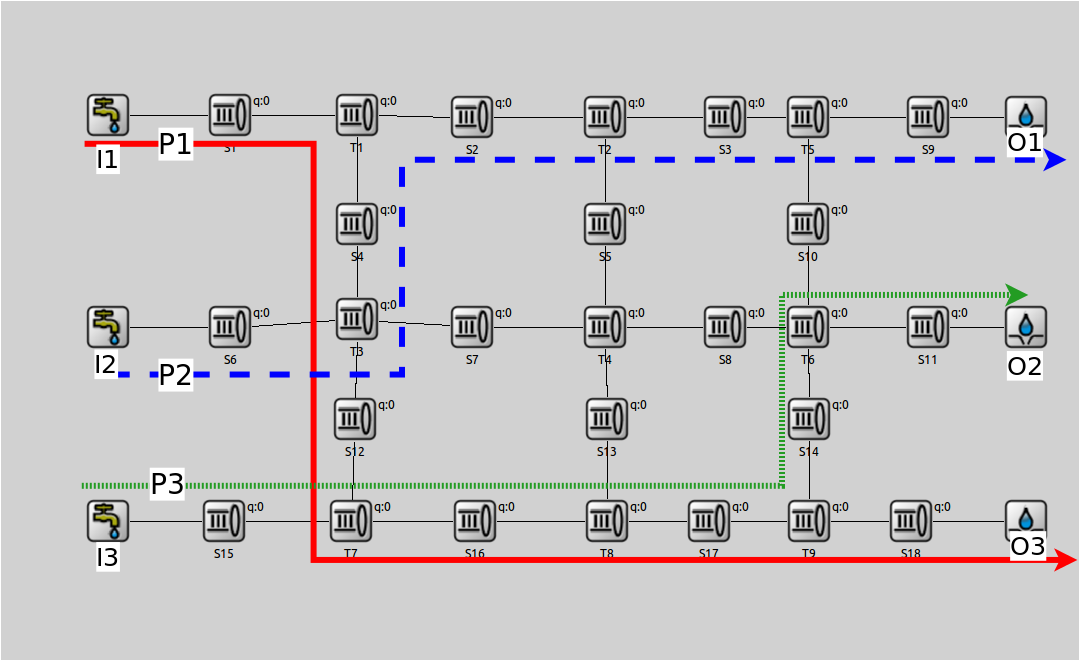}
   \caption{Three-input, three-output conveyor system.}
  \label{fig:3I-3O-topo}
 \end{figure}

 \begin{figure}
  \centering
  \includegraphics[width=0.75\textwidth,height=6.cm]{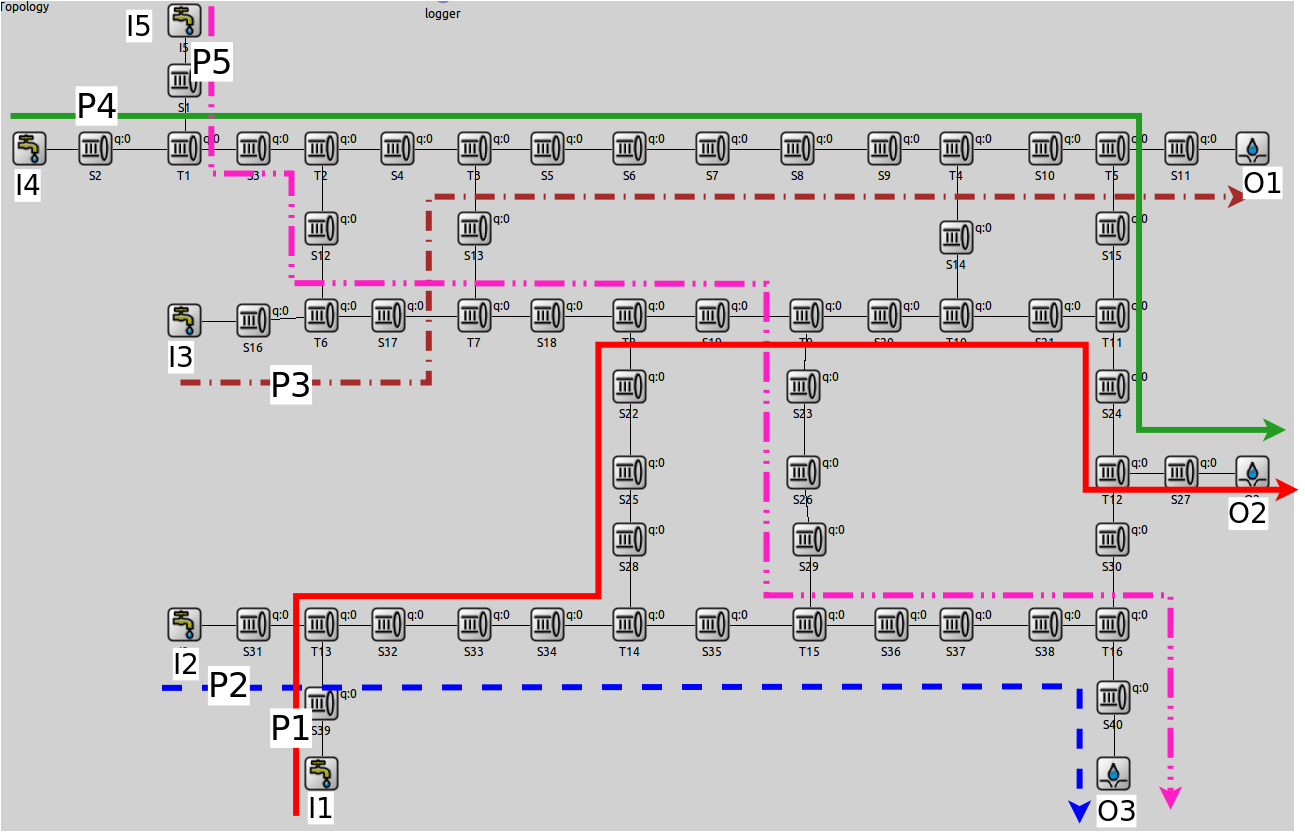}
   \caption{Package Sorter topology.}
  \label{fig:Packet-Sorter-topo}
 \end{figure}
 
\begin{table*}
\small
\footnotesize
	\caption{Various Conveyor system topologies, input rates, alphabet size $M$, number of events (in the final data set with the topology) and 
the paths in each topology.}
	\label{tableTopos}
\begin{adjustwidth}{-.75in}{-.75in}
      \begin{tabular}{@{}ccccl@{}} \toprule
	Topology & InputRate &$M$ &Events & Paths	\\ 
		 & (Poisson) & & &	\\ 
	\midrule
2I-2O	&0.6 & 16& 7659 	&{\bf P1}. I1 S1 T1 S4 T3 S7 T4 S8 O2 \\
	&  &	&	&{\bf P2}. I2 S6 T3 S7 T4 S5 T2 S3 O1 \\
	\midrule
3I-3O	& 0.4 &33& 12497	& {\bf P1}. I1 S1 T1 S4 T3 S12 T7 S16 T8 S17 T9 S18 O3 \\
	&   & &	& {\bf P2}. I2 S6 T3 S7 T4 S5 T2 S3 T5 S9 O1 \\
	&   & &	& {\bf P3}. I3 S15 T7 S16 T8 S17 T9 S14 T6 S11 O2 \\	
	\midrule
Package & 0.2 & 64 &18489 	& {\bf P1}. I1 S39 T13 S32 S33 S34 T14 S28 S25 S22 T8  S19 T9  S20 T10 S21 T11 S24 T12
S27 O2 \\
Sorter	& &&	& {\bf P2}. I2 S31 T13 S32 S33 S34 T14 S35 T15 S36 S37 S38 T16 S40 O3 \\
	& &&	& {\bf P3}. I3 S16 T6  S17 T7  S13 T3  S5  S6  S7  S8  S9  T4  S10 T5  S11 O1 \\
	& &&	& {\bf P4}. I4 S2  T1  S3  T2  S4  T3  S5  S6  S7  S8  S9  T4  S10 T5  S15 T11 S24 T12 S27 O2 \\
	& &&	& {\bf P5}. I5 S1  T1  S3  T2  S12 T6  S17 T7  S18 T8  S19 T9  S23 S26 S29 T15 S36 S37 S38 T16 S40 O3 \\
	\bottomrule
      \end{tabular}
     \end{adjustwidth}
     
\end{table*}

For the conveyor system data, there is only one data sequence corresponding to each topology. The GoKrimp algorithm
works using a statistical dependency test called $Sign Test$, which does not work on single sequence datasets. The
implementation of GoKrimp asks such sequences to be broken into a number smaller sequences. Hence, the datasets were
broken into sequences of size 25, 50 and 100 respectively for the 2I-2O, 3I-3O and the Package Sorter Topologies (creating
longer sizes for sequences with longer and higher number of paths). Also, the conveyor system datasets are time
stamped datasets. The SQS and GoKrimp algorithms do not take time stamped data. Hence experiments for SQS and GoKrimp on
these datasets is carried out after removing the time stamps, while retaining the temporal ordering of the symbols.

\subsubsection{Interpretability of Selected Patterns on conveyor system dataset}
Since the Conveyor system dataset consists of movement of packages through specific paths, we expect the good patterns
to be sub paths of these paths. 
The results on the three Conveyor system data sets is given in Table~\ref{table:patt-conveyor}. In all the three data
sequences, our algorithm retrieves episodes corresponding to the sub paths of package flow (see Table~\ref{tableTopos}
for the actual paths). These patterns actually correspond to sub paths where the packages move along the segments and
turns at a fixed speed, without any congestion. The patterns returned by SQS and GoKrimp did not have any particular significance with respect to the
topologies. Even though SQS gave some long patterns, those did not give any information regarding the underlying system of transportation of packages.
\begin{table*}
\small
\footnotesize
	\caption{Patterns discovered by various algorithm in the Conveyor System data sets. The patterns for the CSC-2
is shown without the inter event gap. PS stands for Package Sorter.}
	\label{table:patt-conveyor}
\begin{adjustwidth}{-.75in}{-.75in}
      \begin{tabular}{@{}cll@{}} \toprule
	Method & Topology &{Top Discovered Patterns}	\\ 
	\midrule
	&2I-2O &[ S8 T4 S4 S2 T1 ], [ T3 I2 S3 S6 ],  [ I1 O2 S8 ],  [ S6 S4 S7 S5 ], [ S2 T1 S1 ],  [ S8 T4 S4 ]
\\[3pt]	
GOKRIMP	&3I-3O & [ T5 T9 ], [ S10 I1 S14 ],  [ S2 T1 ], [ T3 T6 S18 ],  [ S16 S12 ], [ S8 S6 ] \\[3pt]
	&PS & [ T2 S4 ], [ S27 O2 ],  [ S10 S12 ],  [ S17 S18 ],  [ S19 T10 ],  [ S34 S35 ] \\ \midrule

	&2I-2O &[S3 S8 T2 S7 T4 T3 S6 I2 O1 O2 T4], [O1 S8 T4], [O1 S8 T2 T4], \\ 
	&	& [S3 S8 T2 T4 T3 S6 I2 O1 O2], [T4 T3 S6][S5 S7 T4 T3] \\ [3pt]		
SQS	&3I-3O &[T6 T9], [O2 S18], [O1 S11], [O3 S9], [T8 T7 S7], [T7 T3] \\ [3pt]
	&PS &[T16 S24 ], [S15 S36 ], [T11 S36 ], [T5 S24 ], [T10 T15 ], [S21 S36 ], [T12 S38 ]\\
	\midrule

	&2I-2O &[S6 T3 S7 T4 S5 T2 S3 O1], [S1 T1 S4 T3 S7 T4 S8 O2], [S4 T3 S7 T4 S8 O2] \\ 
	&	& [I1 I2], [T3 S7 T4], [T1 I2]\\	[3pt]	
CSC-2	&3I-3O &[S15 T7 S16 T8 S17 T9 S14 T6 S11 O2], [S6 T3 S7 T4 S5 T2 S3 T5 S9 O1]\\
	& 	&[S12 T7 S16 T8 S17 T9 S18 O3], [I1 S1 T1 T3 T7 S16 T8 S17 T9], [T7 S16 T8 S17 T9] \\[3pt]	
	&PS & [T3 S5 S6 S7 S8 S9 T4 S10 T5], [S12 T6 S17 T7 T8 S19 T9 S23 S26 T15 S36 S37 S38 T16 S40 O3]\\
	& & [S39 T13 S32 S33 S34 T14 S28 S25 T8 S19 T9 S20 T10 T11 S24 T12 S27 O2]  \\
& &[T15 S36 S37 S38 T16 S40 O3],  [T13 S32 S33 S34 T14] \\
 & &[S16 T6 S17 T7 T3 S5 S6 S7 S8 S9 T4 S10 T5],  [S1 T1 S3 T2 T6 S17 T7] \\
	\bottomrule
      \end{tabular}
 \end{adjustwidth}
      

\end{table*}

\subsubsection{Effectiveness and Efficiency of the Algorithms}
In this section, we compare the effectiveness and efficiency of the methods in terms of compression achieved, the number of patterns returned  and 
the run time on the conveyor system data set. Since the other methods, GoKrimp and SQS, employ a bit level coding scheme, for comparison
purpose, we also calculate the exact bits needed for our encoding scheme, even though it doesn't affect the encoding efficiency as discussed earlier. 
The encoding scheme remains exactly the same except for the specification of the alphabet size $M$ in the beginning of the encoding. Since we only
deal with positive integers, we assume that all the positive integers are encoded using the Elias codes~\cite{lam2014mining}. For a positive
integer, $n$, the Elias Code Length for encoding $n$ is $2\lfloor \log_2(n)\rfloor +1$. Each event-type is encoded using a fixed code size of
$\lfloor\log M\rfloor
+ 1$ bits. It is easy to see that this bit level encoding scheme can also be easily decoded as discussed earlier. 

For each method, the compression achieved is calculated as the ratio of size of encoded data using singleton patterns to the size of encoded data 
using the selected patterns returned by the methods. Table~\ref{table:compression2} shows the compression ratios achieved by the three algorithms on 
the conveyor system data sets. Here, the column CSC-2 lists compression ratio in terms of memory units, while the last column CSC-2(Bitwise) gives the
bit level compression ratio for CSC-2. As can be seen, our compression ratio is better by a factor of 3. As mentioned earlier, remote visualization 
and monitoring of such composable conveyor system is a potential application area of temporal data mining where the compression achieved is also 
important. Our method achieves better compression and also returns a very relevant set of episodes.

\begin{table}
\centering
 \caption{Compression achieved by various algorithms on conveyor system data sequences. }
 \label{table:compression2}
  \begin{tabular}{@{}ccccc@{}} \toprule
   Datasets 	& SQS & GoKrimp & CSC-2 & CSC-2 \\ 
		&     &		&	& (Bitwise) \\\midrule
   2I-2O 	&1.42 & 1.36	& {\bf4.56} & {\bf4.79} \\
   3I-3O		&1.13 & 1.11	&{\bf4.63}  & {\bf4.89} \\	
   PackageSorter&1.04 &	1.015   &{\bf3.34}  & {\bf3.6} \\
   \bottomrule
  \end{tabular}
\end{table}

Table~\ref{table:numpatterns2} shows the number of patterns returned by the three methods and Table~\ref{table:runtime2} shows run times of the three
methods. Ideally for datasets with some inherent regularity in the sequence, a few significant patterns should explain the dataset well. In such
cases, a good method should return a few relevant patterns. For the conveyor system datasets, where there is inherent regularity in the form of packet
flows, CSC-2 returns a few highly relevant patterns, which happen to be sub paths of packet flows. GoKrimp also returns a small number of
patterns, though as seen in Table~\ref{table:patt-conveyor}, they are not relevant episodes. The PackageSorter topology is a complex topology with
lots of intersecting paths. Hence CSC-2 captures more number of patterns here than from other topologies. As shown in Table~\ref{table:patt-conveyor},
the patterns corresponded to sub paths relevant to the topology. But for the same topology, GoKrimp gives the minimum number of patterns.
However none of these patterns are relevant in the sense that they capture good regularities in the underlying system. Hence GoKrimp algorithm fails
to even see any inherent regularity in the datasets (which is also seen by the low compression achieved). SQS returns the highest number of patterns
in all cases, which are all but irrelevant to these
topologies. 

\begin{table}
  \centering
 \caption{Number of Patterns returned by various algorithms on conveyor system data sequences.}
 \label{table:numpatterns2}
  \begin{tabular}{@{}cccc@{}} \toprule
   Datasets 	& SQS & GoKrimp & CSC-2 \\ \midrule
   2I-2O 	&65 & 20	 & 13 \\
   3I-3O		&114 & 22	 & 27 \\
   PackageSorter&140 & 9	 & 77 \\ 
   
   \bottomrule
  \end{tabular}
\end{table}

\begin{table}
  \centering
 \caption{Run times in seconds on conveyor system data sequences.}
 \label{table:runtime2}
  \begin{tabular}{@{}cccc@{}} \toprule
   Datasets 	& SQS & GoKrimp & CSC-2 \\ \midrule
   2I-2O 	&6 & 1	 & 1 \\
   3I-3O		&14 & 1	 & 1 \\
   PackageSorter&20 & 1 & 2 \\ 
   \bottomrule
  \end{tabular}
\end{table}

\subsection{Results on other data sets}
We now discuss the experimental results on the datasets given in Table~\ref{table:datasummary}. We first show the interpretability of the patterns
output by different methods on the JMLR data. In section~\ref{sec:efficiency-realworld}, we compare the methods for efficiency in terms
of compression achieved, runtime and the number of patterns. In section ~\ref{sec:classification}, the effectiveness of the selected patterns from
different methods is analyzed using the patterns as features for classification. 

\subsubsection{Interpretability of Selected Patterns}
Of the six real world data sets, only the outputs from the JMLR data could be analyzed for interpretability.  For the
JMLR data set, we expect key phrases relevant to Machine Learning Research to pop up. We ran the CSC-2 algorithm with $K = 20$ to find the top 20 
selected patterns. For the other two algorithms we just selected the top 20 of the outputted patterns. Table~\ref{table:patt-jmlr} gives the patterns 
obtained on JMLR text data. As can be seen, the patterns returned by all three of the methods are relevant and almost identical. Also we couldn't see 
any redundancy in any of the pattern sets. 

\begin{table*}
\small
\footnotesize
	\caption{Patterns discovered by various algorithm in the JMLR data.}
	\label{table:patt-jmlr}
\begin{adjustwidth}{-.75in}{-.75in}
      \begin{tabular}{@{}cllll@{}} \toprule
	Method & \multicolumn{4}{c}{Top Discovered Patterns}	\\ 
	\midrule
	&	support vector machin & state art & neural network & well known \\
	&	real world & high dimension & experiment result & special case \\
GOKRIMP	&	machin learn&  reproduc hilbert space & sampl size & solv problem \\
	&	data set &larg scale &supervis learn& signific improv \\
	&	bayesian network & independ compon analysi &support vector &object function \\
	\midrule
	& 	support vector machin &larg scale &featur select &sampl size	\\
        &	machin learn &nearest neighbor &graphic model &learn algorithm \\
SQS     &	state art &decis tree &real world &princip compon analysi	\\
        &	data set &neural network &high dimension &logist regress	\\
        &	bayesian network &cross valid &mutual inform &model select	\\
        \midrule
	&support vector machin &reproduc kernel hilbert space &classif problem & experiment result \\
	&data set &real world &optim problem &gener error	\\
CSC-2	&learn algorithm &model select &loss function &supervis learn \\
	&machin learn &featur select &state art &graphic model \\
	&bayesian network &high dimension &larg scale &propos method \\
	\bottomrule
      \end{tabular}
\end{adjustwidth}
\end{table*}

\subsubsection{Efficiency of the algorithms}
\label{sec:efficiency-realworld}
In this section, we analyze the efficiency of the methods in terms of compression achieved, the number of patterns
returned and the run time on the data sets listed in Table~\ref{table:datasummary}. Table~\ref{table:compression}
presents the comparison of compression ratio for different methods. Again, the column named CSC-2 lists compression ratio in
terms of memory units and the last column, CSC-2(Bitwise), gives the bit level compression ratio for CSC-2.
As is seen in Table~\ref{table:compression}, CSC-2 offers better compression on almost all the datasets. 

\begin{table}
  \centering
 \caption{Compression achieved by various algorithms.}
 \label{table:compression}
  \begin{tabular}{@{}ccccc@{}} \toprule
   Datasets 	& SQS & GoKrimp & CSC-2 & CSC-2 \\ 
		&     &		&	& (Bitwise) \\\midrule
   jmlr		&1.039& 1.008   & {\bf1.07}  & {\bf1.117}\\ 
   aslbu		&1.155& 1.123 	& {\bf1.17} & {\bf1.24}\\
   aslgt		&1.308& 1.156	& {\bf1.98} & {\bf1.99}\\
   auslan2	&1.571&1.428	& {\bf1.88} & {\bf1.96} \\
   context	&{\bf2.7} & 1.7	&1.95 &1.98 \\
   pioneer	&1.3 & 1.17	&{\bf1.58} & {\bf1.74} \\
   \bottomrule
  \end{tabular}
\end{table}

Table~\ref{table:numpatterns} shows the number of patterns returned by the three methods on real data. For the JMLR data, GoKrimp extracted only 20
patterns, which as we showed in Table~\ref{table:patt-jmlr}, are all relevant. For SQS and CSC-2, the number of selected patterns were quite large.
The initial phrases were all relevant to machine learning, as seen in Table~\ref{table:patt-jmlr}. Later on, for the CSC-2 algorithm, most of the
patterns returned were either relevant or slightly relevant machine learning, optimization and mathematical phrases with a bit of redundancy between
the phrases. Never the less, we could make some sense out of most of the returned phrases.
For the rest of the data sets, the size of the selected set of CSC-2 is in between that of the GoKrimp and SQS methods. We would like to
point that the size of the selected set as such, without considering the relevance of the selected patterns does not have any particular significance.
In Section~\ref{sec:classification}, we discuss this again in relation to classification accuracy.

\begin{table}
  \centering
 \caption{Number of Patterns returned by various algorithms.}
 \label{table:numpatterns}
  \begin{tabular}{@{}cccc@{}} \toprule
   Datasets 	& SQS & GoKrimp & CSC-2 \\ \midrule
   jmlr		&580 & 20 & 765 \\
   aslbu		&195 & 67 & 453 \\
   aslgt		&1095 & 68 & 105 \\
   auslan2	&13 & 4	 & 7 \\
   context	&138 & 33 & 39 \\
   pioneer	&143 &49 & 134\\
   \bottomrule
  \end{tabular}
\end{table}

Table~\ref{table:runtime} compares the run times. Except for the JMLR text data, our method took less than 2s
for all the other data sets. For the JMLR data, the number of event types $M$ is high (see
Table~\ref{table:datasummary}) and hence the candidate set $\scrC$ for CSC-2 is high, which results in more time for
calculating the $OM$ matrix. On the data set {\it aslgt}, both SQS and GoKrimp take a lot of time because the number of
events in the sequence is very large. However this does not affect our method much.
\begin{table}
  \centering
 \caption{Run times in seconds.}
 \label{table:runtime}
  \begin{tabular}{@{}cccc@{}} \toprule
   Datasets 	& SQS & GoKrimp & CSC-2 \\ \midrule
   jmlr		&489 & {\bf2} & 182 \\ 
   aslbu		&196 & 28 & {\bf5} \\
   aslgt		&$>10$Hrs & 1440 & {\bf4} \\
   auslan2	&1 & 1	 & 1 \\
   context	&70 & 36 & {\bf1} \\
   pioneer	&10 &6 & {\bf1}\\
   \bottomrule
  \end{tabular}
\end{table}

\subsubsection{Classification Results}
\label{sec:classification}
Here we show the classification results on the five labeled datasets in Table~\ref{table:datasummary} (The JMLR data has no class labels). For each
method, we select {\em class specific} patterns using the respective methods and merge them along with the set of all singletons (which are 1-node
episodes) to form features/attributes for classification. Thus, for classification, each sequence is represented by a feature vector consisting of the
number of occurrences of each of the (selected) patterns together with the counts of the singletons in the sequence . 
We also show the results, when only singletons are used as features. We call this method as {\it Singletons} in the tables. 
In~\cite{lam2014mining}, GoKrimp was compared with closed mining episode method, BIDE~\cite{wang2004bide}, and one of
their two step (mining and then subset selection) methods, SEQKrimp~\cite{lam2014mining, lam2012mining}. They used different classifiers  and found
that the linear SVM classifier was giving the best results for almost all datasets. And for all the datasets, they showed that the top patterns
returned by the SeqKrimp and GoKrimp were giving better classification accuracy than the top patterns returned by the BIDE algorithm. For our study,
we thus use the linear SVM as the classifier using the libsvm package~\cite{CC01a}. Since SeqKrimp and GoKrimp were giving similar results, we only
use GoKrimp for this comparison study as SeqKrimp is computationally very intensive. 

The fixed interval serial episode patterns returned by CSC-2 are used as features in two different ways. In the first method, features are the counts
of the selected fixed interval serial episodes along with the counts of the singletons. In the second approach, we drop the fixed inter-event
constraints from the episodes and treat them as normal serial episodes (and care is taken to avoid multiple representations of the same serial
episode). The counts of these serial episodes along with the singleton counts would be the feature representation of the sequences. The first method
is called as {\it CSC-2} and the second method is called {\it CSC-2$^*$} in the results table.

%

\begin{table}
  \centering
 \caption{Percentage Classification accuracy achieved by various algorithms on 20 runs.}
 \label{table:class-accuracy}
  \begin{tabular}{@{}ccccccccccc@{}} \toprule
  Datasets & \multicolumn{2}{c}{SQS}& \multicolumn{2}{c}{GoKrimp} & \multicolumn{2}{c}{CSC-2} & \multicolumn{2}{c}{CSC-2$^{*}$} & 
\multicolumn{2}{c}{Singletons} \\
 & \multicolumn{2}{c}{Accuracy} & \multicolumn{2}{c}{Accuracy} & \multicolumn{2}{c}{Accuracy} & \multicolumn{2}{c}{Accuracy} &
\multicolumn{2}{c}{Accuracy} \\
  & Mean & $\sigma$ & Mean & $\sigma$ & Mean & $\sigma$ & Mean & $\sigma$ & Mean & $\sigma$ \\

  \midrule
   aslbu		&70.08 & 0.68 	& 68.57 	& 0.83 	&70.02 	& 0.94 	&{\bf71.15} 	& 0.75 	& 69.7 	& 1.1	\\
   aslgt		&{\bf86.12} & 0.18	& 85.55	& 0.23 	&85.9 	& 0.15 	&{\bf86.1} 	& 0.19 	& 82.17 & 0.19	\\
   auslan2	&{\bf35}    & 1.23	& 32.7	& 1.6 	&{\bf35.18}  & 1.2 	&34.05  & 1.13 	& 32.58 	& 1.7	\\
   context	&{\bf94.02} & 0.94	& 93.83	& 0.61 	&93.7 	& 0.89 	&{\bf94} 	& 0.65 	& {\bf94.02} & 0.58	\\
   pioneer	&100   & 0	& 100	& 0 	&100   	& 0 	&100   	& 0 	& 100 	& 0	\\
   \bottomrule
  \end{tabular}
\end{table}

\begin{table}
  \centering
 \caption{The average number of features per data set for each method. The number (except for the Singletons) is the number of non-singleton
extracted patterns which represent the feature set along with the singletons.}
 \label{table:classify_numpatterns}
  \begin{tabular}{@{}ccccc@{}} \toprule
   Datasets 	& SQS & GoKrimp & CSC-2 & Singletons \\ \midrule
   aslbu		&70 & 17 & 23 & 263\\
   aslgt		&1154 & 418 & 923 & 94 \\
   auslan2	&15 & 0	 & 6 & 23\\
   context	&150 & 64 & 42 & 107\\
   pioneer	&132 &24 & 121 & 184\\
   \bottomrule
  \end{tabular}
\end{table}
For each dataset, the results were obtained by averaging 20 repetitions of 10-fold cross-validation. The results are shown in
Table~\ref{table:class-accuracy}. The table gives the mean and standard deviation ($\sigma$) of the classification accuracy. 
We see that, SQS marginally outperforms the other methods for the {\it aslgt} and {\it context} datasets and the two CSC-2 methods have higher
accuracies for the {\it aslbu} and {\em auslan2} data sets, respectively. In all the datasets, however, the accuracies are similar for SQS and CSC-2.
In comparison, the accuracies of the GoKrimp method are slightly lower than both the SQS and CSC-2 methods, except for the pioneer
data, where none of the methods seem to misclassify.
Table~\ref{table:classify_numpatterns} shows the number of patterns selected by each method as features and we see that the number of features
selected by SQS is far higher than other methods. And we noticed that, even though the performance of SQS was only slightly better than CSC-2,
the run time for the experiments was at least five times longer than that of CSC-2. On the other hand, GoKrimp selects comparatively the lowest number
of patterns (except for the context data set) and also has the lowest accuracy among the three methods. For the auslan2 dataset, GoKrimp couldn't
extract any pattern and hence the classification accuracy is similar to $Singletons$ (the slight difference is due to difference in the Cross
Validation splits for different runs).
It is also interesting to see that the $Singletons$ method, which consists of only the 1-node counts are always close to the best results. But
nevertheless, the table shows that the selected patterns
from different methods have contributed to the increase in accuracy. 

\subsection{Discussion}

The results presented here show that CSC-2 is a good method for finding a subset of patterns that achieve good compression. It is also seen that
these patterns that achieve good compression are also highly relevant for the problem. For the conveyor system
datasets, our method was shown to perform extremely well in pulling out patterns representing the stable flow of items and achieving great
compression. Both the aspects are of great importance for remote monitoring of such systems as discussed earlier. The other methods have
failed in doing so. On these datasets, the other algorithms, namely SQS and GoKrimp failed to find patterns that capture the package flow and they
also could not achieve much data compression.  

We have also tested out method on some real world data sets. The CSC-2 algorithm discovered subsets of episodes that result in better data
compression compared to the  other methods and our algorithm also seems faster than other algorithms for most of the datasets. The subset of patterns
identified by CSC-2 is also seen to be very effective in classification scenarios. 

Recall that the patterns used by CSC-2 are fixed interval serial episodes. Such episodes, as we have seen, suited the conveyor system
data sets, where sequential occurrences of events follow such a fixed gap mechanism. For the other sequential datasets used for classification, such
constraints may not be really relevant. But even on these data sets, CSC-2 identifies a subset of patterns that result in both data compression as
well as better performance in classification. This shows that our pattern structure is not particularly restrictive and it is useful on a variety of
data sets. 

\section{Conclusion}
\label{sec:conclude}
Frequent episodes discovered from sequential data are supposed to give us good insights into the characteristics of the 
data source. However, in practice, most mining algorithms output a large number of highly redundant episodes. Isolating a small 
subset of episodes that succinctly characterize the data is a challenging problem. 
In this paper, we presented an MDL based approach for this problem. Using the interesting class of fixed interval serial
episodes and a novel data encoding scheme, we presented a method to discover a subset of highly relevant episodes.
In contrast to methods in~\cite{lam2012mining,tatti2012long,lam2014mining}, our method achieves good data compression,
while being able to work with event sequences with time stamps.

We compared our method with SQS and GoKrimp on text data and also on a number of real world data sets which were used
earlier in temporal data mining. On all these data sets, our method is good in comparison to others, both in terms of
compression and run time. For the classification scenario, our method was only slightly less effective than SQS but better than GoKrimp. But we
achieved it with far fewer patterns and very low run times.

In this paper, we also briefly discussed a novel application area for sequential pattern mining. This is the
composable conveyor system. We presented empirical comparison of our method with that of others on three data sets from
this problem domain to demonstrate both the effectiveness and efficiency of our method. 

In this paper, we have not attempted any statistical analysis of our method so that we can relate the data compression to some measure of statistical
significance of the pattern subset isolated by our method. This is an interesting and challenging direction to extend the work presented here. We
would be exploring this in our future work.


\begin{thebibliography}{10}

\bibitem{avinash2013pattern}
A.~Achar, A.~Ibrahim, and P.~S. Sastry.
\newblock Pattern-growth based frequent serial episode discovery.
\newblock {\em Data \& Knowledge Engineering}, 87:91--108, 2013.

\bibitem{achar2012unified}
A.~Achar, S.~Laxman, and P.~S. Sastry.
\newblock A unified view of the apriori-based algorithms for frequent episode
  discovery.
\newblock {\em Knowledge and information systems}, 31(2):223--250, 2012.

\bibitem{asrr2009}
B.~Archer, S.~Shivakumar, A.~Rowe, and R.~Rajkumar.
\newblock Profiling primitives of networked embedded automation.
\newblock In {\em Automation Science and Engineering, 2009. CASE 2009. IEEE
  International Conference on}, pages 531--536. IEEE, 2009.

\bibitem{burdick2005mafia}
D.~Burdick, M.~Calimlim, J.~Flannick, J.~Gehrke, and T.~Yiu.
\newblock Mafia: A maximal frequent itemset algorithm.
\newblock {\em Knowledge and Data Engineering, IEEE Transactions on},
  17(11):1490--1504, 2005.

\bibitem{calders2002mining}
T.~Calders and B.~Goethals.
\newblock Mining all non-derivable frequent itemsets.
\newblock In {\em Principles of Data Mining and Knowledge Discovery}, pages
  74--86. Springer, 2002.

\bibitem{casas2005summarizing}
G.~Casas-Garriga.
\newblock Summarizing sequential data with closed partial orders.
\newblock In {\em SDM}, volume~5, pages 380--391. SIAM, 2005.

\bibitem{chandola2007summarization}
V.~Chandola and V.~Kumar.
\newblock Summarization--compressing data into an informative representation.
\newblock {\em Knowledge and Information Systems}, 12(3):355--378, 2007.

\bibitem{CC01a}
C.-C. Chang and C.-J. Lin.
\newblock {LIBSVM}: A library for support vector machines.
\newblock {\em ACM Transactions on Intelligent Systems and Technology},
  2:27:1--27:27, 2011.
\newblock Software available at \url{http://www.csie.ntu.edu.tw/~cjlin/libsvm}.

\bibitem{geerts2004tiling}
F.~Geerts, B.~Goethals, and T.~Mielik{\"a}inen.
\newblock Tiling databases.
\newblock In {\em Discovery science}, pages 278--289. Springer, 2004.

\bibitem{grunwald2007minimum}
P.~D. Gr{\"u}nwald.
\newblock {\em The minimum description length principle}.
\newblock The MIT Press, 2007.

\bibitem{han2007frequent}
J.~Han, H.~Cheng, D.~Xin, and X.~Yan.
\newblock Frequent pattern mining: current status and future directions.
\newblock {\em Data Mining and Knowledge Discovery}, 15(1):55--86, 2007.

\bibitem{lam2012mining}
H.~T. Lam, F.~M{\"o}rchen, D.~Fradkin, and T.~Calders.
\newblock Mining compressing sequential patterns.
\newblock In {\em SDM}, pages 319--330. SIAM, 2012.

\bibitem{lam2014mining}
H.~T. Lam, F.~M{\"o}rchen, D.~Fradkin, and T.~Calders.
\newblock Mining compressing sequential patterns.
\newblock {\em Statistical Analysis and Data Mining}, 7(1):34--52, 2014.

\bibitem{laxman2007fast}
S.~Laxman, P.~S. Sastry, and K.~P. Unnikrishnan.
\newblock A fast algorithm for finding frequent episodes in event streams.
\newblock In {\em Proceedings of the 13th ACM SIGKDD international conference
  on Knowledge discovery and data mining}, pages 410--419. ACM, 2007.

\bibitem{lin2002pincer}
D.-I. Lin and Z.~M. Kedem.
\newblock Pincer-search: an efficient algorithm for discovering the maximum
  frequent set.
\newblock {\em Knowledge and Data Engineering, IEEE Transactions on},
  14(3):553--566, 2002.

\bibitem{mannila1997discovery}
H.~Mannila, H.~Toivonen, and A.~I. Verkamo.
\newblock Discovery of frequent episodes in event sequences.
\newblock {\em Data Mining and Knowledge Discovery}, 1(3):259--289, 1997.

\bibitem{MR04}
N.~M{\'e}ger and C.~Rigotti.
\newblock Constraint-based mining of episode rules and optimal window sizes.
\newblock {\em Knowledge Discovery in Databases: PKDD 2004}, pages 313--324,
  2004.

\bibitem{moerchen2010robust}
F.~Moerchen and D.~Fradkin.
\newblock Robust mining of time intervals with semi-interval partial order
  patterns.
\newblock In {\em SDM}, pages 315--326. SIAM, 2010.

\bibitem{pasquier1999discovering}
N.~Pasquier, Y.~Bastide, R.~Taouil, and L.~Lakhal.
\newblock Discovering frequent closed itemsets for association rules.
\newblock In {\em Database Theory—ICDT’99}, pages 398--416. Springer, 1999.

\bibitem{rissanen1983universal}
J.~Rissanen.
\newblock A universal prior for integers and estimation by minimum description
  length.
\newblock {\em The Annals of statistics}, pages 416--431, 1983.

\bibitem{sks2006}
S.~Shivakumar.
\newblock Sensor-actuator systems for automation.
\newblock In {\em Work In Progress Session, IEEE Real-time Systems Symposium}.
  IEEE, 2006.

\bibitem{siebes2006item}
A.~Siebes, J.~Vreeken, and M.~van Leeuwen.
\newblock Item sets that compress.
\newblock In {\em SDM}, volume~6, pages 393--404. SIAM, 2006.

\bibitem{tatti2012long}
N.~Tatti and J.~Vreeken.
\newblock The long and the short of it: Summarising event sequences with serial
  episodes.
\newblock In {\em Proceedings of the 18th ACM SIGKDD international conference
  on Knowledge discovery and data mining}, pages 462--470. ACM, 2012.

\bibitem{vreeken2011krimp}
J.~Vreeken, M.~Van~Leeuwen, and A.~Siebes.
\newblock Krimp: mining itemsets that compress.
\newblock {\em Data Mining and Knowledge Discovery}, 23(1):169--214, 2011.

\bibitem{wang2004bide}
J.~Wang and J.~Han.
\newblock Bide: Efficient mining of frequent closed sequences.
\newblock In {\em Data Engineering, 2004. Proceedings. 20th International
  Conference on}, pages 79--90. IEEE, 2004.

\bibitem{wang2003closet+}
J.~Wang, J.~Han, and J.~Pei.
\newblock Closet+: Searching for the best strategies for mining frequent closed
  itemsets.
\newblock In {\em Proceedings of the ninth ACM SIGKDD international conference
  on Knowledge discovery and data mining}, pages 236--245. ACM, 2003.

\bibitem{wang2006efficiently}
J.~Wang and G.~Karypis.
\newblock On efficiently summarizing categorical databases.
\newblock {\em Knowledge and Information Systems}, 9(1):19--37, 2006.

\bibitem{witten1999managing}
I.~H. Witten, A.~Moffat, and T.~C. Bell.
\newblock {\em Managing gigabytes: compressing and indexing documents and
  images}.
\newblock Morgan Kaufmann, 1999.

\bibitem{xiang2008succinct}
Y.~Xiang, R.~Jin, D.~Fuhry, and F.~F. Dragan.
\newblock Succinct summarization of transactional databases: an overlapped
  hyperrectangle scheme.
\newblock In {\em Proceedings of the 14th ACM SIGKDD international conference
  on Knowledge discovery and data mining}, pages 758--766. ACM, 2008.

\bibitem{yan2003clospan}
X.~Yan, J.~Han, and R.~Afshar.
\newblock Clospan: Mining closed sequential patterns in large datasets.
\newblock In {\em Proceedings of SIAM International Conference on Data Mining},
  pages 166--177. SIAM, 2003.

\end{thebibliography}

\end{document}